\documentclass[11pt]{llncs}
\usepackage{amsmath}
\usepackage{amssymb}
\usepackage{enumitem}
\usepackage{tikz}

\setlist{nolistsep}

\title{On the Minimum Degree up to Local Complementation: Bounds and Complexity}

\author{J\'er\^ome Javelle\inst{2}, Mehdi Mhalla\inst{1,2} and Simon Perdrix\inst{1,2}}

\institute{CNRS \and Laboratoire d'Informatique de Grenoble, Grenoble University}

\date{\today}

\begin{document}

\maketitle

\sloppy

\begin{abstract}
The local minimum degree of a graph is the minimum degree reached by means of a series of local complementations.
In this paper, we investigate on this quantity which plays an important role in quantum computation and quantum error correcting codes.

First, we show that the local minimum degree of the Paley graph of order $p$ is greater than $\sqrt{p} - \frac{3}{2}$, which is, up to our knowledge, the highest known bound on an explicit family of graphs.
Probabilistic methods allows us to derive the existence of an infinite number of graphs whose local minimum degree is linear in their order with constant $0.189$ for graphs in general and $0.110$ for bipartite graphs.
As regards the computational complexity of the decision problem associated with the local minimum degree, we show that it is NP-complete and that there exists no $l$-approximation algorithm for this problem for any constant $l$ unless $P = NP$.
\end{abstract}

\section{Introduction}

For any undirected graph $G$, the local complementation is an operation which consists in complementing the neighborhood of a given vertex of a graph. It that has been introduced by Kotzig \cite{kotzig} and the study of this quantity is motivated by several applications: Bouchet \cite{Bouchet,Bouchet_circle_graphs} and de Fraysseix \cite{HdF} used local complementation to give a characterization of circle graphs, and Oum \cite{Oum_vertex_minor} links the notion of ``vertex minor of a graph" to the equivalence classes up to local complementation.
One of the most important results is established by Bouchet in \cite{Bouchet_LCeq_poly}: deciding whether two graphs are equivalent up to local complementations can be done in polynomial time.

In the field of quantum information theory, the rate of some quantum codes obtained by graph concatenation can be bounded by the minimum degree up to local complementation (called ``local minimum degree" and denoted $\delta_{loc}$) of the constructed graphs \cite{graph_concatenation}.
Another application of $\delta_{loc}$ is the preparation of graph states (quantum states represented by a graph), which are a very powerful tool used for measurement-based quantum computing \cite{BRQC} and blind quantum computing \cite{blind}, for example.
In \cite{GSprep}, it has been proven that the complexity of preparation of a graph state is bounded by its local minimal degree.
Threshold quantum secret sharing protocols from graph states (first introduced in \cite{MS}) can be built from graph states with the methods described in \cite{QSS_GS}, and the local minimum degree of the corresponding graphs gives, under additional parity conditions, a value for the threshold that can be reached with these graph states.
Moreover, we also focus on bipartite graphs which are of high interest for entanglement purification \cite{ent_2_color} and the study of Schmidt measure \cite{Schmidt_measure}, for example.

In this paper, several techniques from different backgrounds are used.
We consider a family of graphs defined from quadratic residues, the Paley graphs $Pal_p$, and the bound that we give on $\delta_{loc}(Pal_p)$ is closely related to a fundamental result in algebraic geometry (see Lemma \ref{Joyner_bound}).
Probabilistic methods are also used to prove the existence of graphs with large local minimum degree.
In particular, we use the asymmetric version of the Lov\'asz Local Lemma \cite{LLL} (see Lemma \ref{LLL}) to prove the existence of an infinite family of graphs with linear $\delta_{loc}$.
We also use this family to derive a polynomial reduction to a problem from coding theory in order to find the computational complexity of finding the local minimum degree of a graph in the general case.

In section \ref{def+Paley}, we recall the definition of the local minimum degree, main notion of this paper, and we give an explicit family of graphs $Pal_p$ of order $p$ such that $\delta_{loc}(Pal_p) \geq \sqrt{p} - \frac{3}{2}$, which is, up to our knowledge, the best known lower bound for any family of graphs.
The next section is dedicated to the proof of the existence of graphs with linear $\delta_{loc}$.
In the last section, we show that the decision problem associated with $\delta_{loc}$ is NP-complete even on the family of bipartite graphs, and we show that there exists no approximation algorithm up to a constant factor for this problem unless $P=NP$.

\section{Definitions}
\label{def+Paley}

Local complementation is defined as follows:

\begin{definition}
The local complementation of a graph $G$ with respect to one of its vertices $u$ results in a graph $G * u = G \Delta K_{\mathcal{N}(u)}$ where $\Delta$ stands for the symmetric difference between edges and $K_{\mathcal{N}(u)}$ is the complete graph on the neighbors of $u$.
\end{definition}

The transitive closure of a graph with respect to the local complementation forms an equivalence class.
In \cite{Bouchet_LCeq_poly}, Bouchet gives a polynomial algorithm that tells whether any two graphs are in the same equivalence class with respect to local complementation.
For a given graph $G$, the quantity we will focus on is the minimum degree of the graphs in its equivalence class.
This value is called the local minimum degree and is written $\delta_{loc}(G)$.
Its formal definition follows:

\begin{definition}
\label{def_delta_loc}
Given a graph $G$,
$\delta_{loc}(G) = min \left\{ \ \delta(G') \ \big| \ G \equiv _{LC}G' \ \right\}$
where $\delta(G')$ is the minimal degree of $G'$ and the equivalence relation $G_1 \equiv_{LC} G_2$ is verified when $G_1$ can be changed into $G_2$ by a series of local complementations.
\end{definition}

In \cite{GSprep}, a characterization of the quantity $\delta_{loc}$ has been established by means of the odd and even neighborhoods of subsets of vertices of a graph defined as follows:

\begin{definition}
Let $G$ be an undirected graph and $D$ a subset of its vertices.
\begin{align}
Odd(D) &= \big\{ \ v \in V(G) \ \big| \ |\mathcal{N}(v) \cap D| = 1 \bmod 2 \ \big\} \\
Even(D) &= \big\{ \ v \in V(G) \ \big| \ |\mathcal{N}(v) \cap D| = 0 \bmod 2 \ \big\}
\end{align}
\end{definition}

The local minimum degree is related to the size of the smallest set of the form $D \cup Odd(D)$:

\begin{property}[\cite{GSprep}]
\label{char_delta_loc}
Let $G$ be an undirected graph.
\begin{align}
\delta_{loc}(G) &= \min \left\{ \ |D \cup Odd(D)| \ \big| \ D \neq \varnothing, D \subseteq V(G) \ \right\} -1
\end{align}
\end{property}

\section{Local minimum degree of Paley graphs}

It is challenging to find a family of graphs with ``high" local minimum degree.
The family of hypercubes, for example, has a logarithmic local minimal degree \cite{GSprep}.

In the following, we prove that a Paley graph of order $n$ has a $\delta_{loc}$ greater than $\sqrt{n}$.
This value is only a lower bound, and we do not know whether it is reached.
This family is defined with quadratic residues over a finite field.
Up to our knowledge, there is no known family of graphs whose local minimum degree is greater than the square root of their order.

For any prime $p$ such that $p = 1 \bmod 4$, the Paley graph $Pal_p$ is a graph on $p$ vertices where each vertex is an element of $\mathbb{F}_p$.
There is an edge between two vertices $i$ and $j$ if and only if $i-j$ is a square in $\mathbb{F}_p$.

\begin{theorem}
\label{delta_loc_Paley}
For any prime $p = 1 \bmod 4$,
\begin{align}
\delta_{loc}(Pal_p) &\geq \sqrt{p} - \frac{3}{2}
\end{align}
where $Pal_p$ is the Paley graph of order $p$.
\end{theorem}

The rest of this section is dedicated to the proof of Theorem \ref{delta_loc_Paley}.
To this end, we give a bound on the size of the sets of the form $D \cup Odd(D)$ in Paley graphs.
The size of such sets is characterized as follows:

\begin{lemma}
\label{odd-even}
For any non-empty set $S \subseteq V(Pal_p)$ and any $i \in V(Pal_p)$,
\begin{align}
\left| \sum_{i=0}^{p-1} \chi_L \left( f_S(i) \right) \right| &= \big| \left| S \cup Odd(S) \right| - \left| S \cup Even(S) \right| \big|
\end{align}
where $f_S(i) = \prod_{j \in S} (i-j)$ and $\chi_L$ is the Legendre character ($\chi_L(x) = x^{\frac{p-1}{2}} \bmod p$).
\end{lemma}

\begin{proof}
First, note that $\chi_L(0) = 0$, $\chi_L(x) = 1$ if $x$ is a quadratic residue in $\mathbb{F}_p$ and $\chi_L(x) = -1$ otherwise.
Since the Legendre character is multiplicative, $\left| \sum_{i=0}^{p-1} \chi_L \left( f_S(i) \right) \right| = \left| \sum_{i=0}^{p-1} \prod_{j \in S} \chi_L(i-j) \right|$.
If $i \in S$ the quantity $\prod_{j \in S} \chi_L(i-j)$ equals $0$.
Otherwise, the product equals $(-1)^{|S|-|\mathcal{N}(i) \cap S|}$, which is $(-1)^{|S|}$ if $i \in Even(S) \setminus S$ and $-(-1)^{|S|}$ if $i \in Odd(S) \setminus S$.
Then, the sum over all vertices $i$ is the difference between the exclusive odd and even neighborhood of the set $S$: $\left| \sum_{i=0}^{p-1} \prod_{j \in S} \chi_L(i-j) \right| = \big| \left| Odd(S) \setminus S \right| - \left| Even(S) \setminus S \right| \big|$.
The last expression can be written $\big| \left| S \cup Odd(S) \right| - \left| S \cup Even(S) \right| \big|$.
\hfill $\Box$
\end{proof}

A well-known result from algebraic geometry related to the hyperelliptic curve of equation $y^2 = \prod_{j \in S} (x-j)$ can be found in \cite{Weil} or \cite{Schmidt}, for example, and is reformulated by Joyner in \cite{JoynerHEC}:

\begin{lemma}[\cite{JoynerHEC}, Proposition 1]
\label{Joyner_bound}
For any non-empty set $S \subseteq \mathbb{F}_p$, let $f_S(x) = \prod_{j \in S} (x-j)$. Then
\begin{align}
\left| \sum_{i \in \mathbb{F}_p} \chi_L \left( f_S(i) \right) \right| &\leq (|S|-1) \sqrt{p} +1
\end{align}
\end{lemma}

This allows us to derive a bound on the sets of type $S \cup Odd(S)$ and $S \cup Even(S)$ in Paley graphs.

\begin{lemma}
\label{bound_odd_even}
Let $Pal_p$ be the Paley graph of order $p$.
For all $S \subseteq V(P_p)$, $S \neq \varnothing$, we have $\sqrt{p} - \frac{1}{2} \leq \left| S \cup Odd(S) \right|$ and $\sqrt{p} - \frac{1}{2} \leq \left| S \cup Even(S) \right|$.
\end{lemma}

\begin{proof}
We consider the case $\left| S \cup Odd(S) \right| \leq \left| S \cup Even(S) \right|$, the other case can be treated a similar way.
Lemma \ref{odd-even} states that $\left| S \cup Odd(S) \right| - \left| S \cup Even(S) \right| = -\left| \sum_{i \in \mathbb{F}_p} \chi_L(f_S(i)) \right|$.
On the other hand, the equality $\left| S \cup Odd(S) \right| + \left| S \cup Even(S) \right| = p+|S|$ is always true.
Thus adding both equalities, $p + |S| - \left| \sum_{i \in \mathbb{F}_p} \chi_L(f_S(i)) \right| = 2 \left| S \cup Odd(S) \right|$.
Thanks to Lemma \ref{Joyner_bound}, we derive $p+|S|-(|S|-1)\sqrt{p}-1 \leq 2 \left| S \cup Odd(S) \right|$.

If $|S| \leq \sqrt{p}$ then the left-hand side of the previous inequality can be bounded: $p+|S|-(|S|-1)\sqrt{p}-1 = p+|S|(1-\sqrt{p})+\sqrt{p}-1 \geq  2 \sqrt{p} - 1$.
Thus, $\sqrt{p} - \frac{1}{2} \leq \left| S \cup Odd(S) \right|$, otherwise $|S| > \sqrt{p}$ and the previous inequality is obviously true.
\hfill $\Box$
\end{proof}
{\bf Proof of Theorem \ref{delta_loc_Paley}:}
The characterization given by Property \ref{char_delta_loc} and the bounds on the size of sets of the form $D \cup Odd(D)$ obtained in Lemma \ref{bound_odd_even} imply that the local minimum degree for Paley graphs is greater than the square root of the order of the graph.
This ends the proof of Theorem \ref{delta_loc_Paley}.

It is significant and interesting to notice that the conjecture of the existence of an infinite family of Paley graphs with linear $\delta_{loc}$ is equivalent to the Bazzi-Mitter conjecture \cite{Bazzi-Mitter}.
However, it is already known that not all Paley graphs have a linear $\delta_{loc}$: there exists no $p_0 \in \mathbb{N}$ such that for all $p>p_0$, $\delta_{loc}(Pal_p)$ is linear in $p$ thanks to Theorem $7$ of \cite{JoynerHEC}.

\section{Existence of graphs with linear local minimum degree}

In this section, we give a proof of the existence of bipartite graphs for which the local minimum degree is linear in the order of the graph.
The proof uses the asymmetric version of Lov\'asz Local Lemma \cite{LLL}:

\begin{lemma}[Asymmetric Lov\'asz Local Lemma]
\label{LLL}
Let $\mathcal{A} = \{A_1, \cdots, A_n\}$ be a set of ÒbadÓ events in an arbitrary probability space and let $\Gamma(A)$ denote a subset of $\mathcal{A}$ such that $A$ is independent from all the events outside $A$ and $\Gamma(A)$.
If for all $A_i$ there exists $\sigma(A_i) \in [0,1)$ such that
$Pr(A_i) \leq \sigma(A_i) \prod_{B_j \in \Gamma(A_i)} (1-\sigma(B_j))$
then we have
$Pr(\overline{A_1}, \cdots, \overline{A_n}) \geq \prod_{A_j \in \mathcal{A}} (1-\sigma(A_j))$.
\end{lemma}

We apply the Local Lov\'asz Lemma (Lemma \ref{LLL}) on random bipartite graphs to show the existence of bipartite graphs with linear local minimum degree.
 \begin{theorem}
 \label{linear_bipartite}
 There exists $\nu_0 \in \mathbb{N}$ such that for all $\nu > \nu_0$ there exists a bipartite graph of order $n=2\nu$ whose local minimum degree is greater than $0.110n$.
 \end{theorem}
 
 \begin{proof}
 Let $G_B$ be a bipartite graph of order $n=2\nu$ with two independent sets of size $\nu$ and where any possible edge exists with probability $\frac{1}{2}$.
 An event which implies that a graph $G$ has a linear $\delta_{loc}$ is: ``$\forall D \subseteq V(G), |D \cup Odd(D)| > cn$" for some $c \in \left] 0,1\right]$.
 In the case of $G_B$, it is sufficient to verify the previous event for sets $D$ such that $D \subseteq V_1$ or $D \subseteq V_2$.
 Indeed, $G_B$ is bipartite, therefore $|D \cup Odd(D)| \geq |(D \cap V_1) \cup Odd(D \cap V_1)|$.
 Therefore we consider the ``bad" events $A^1_D$ and $A^2_D$ defined as follows: if $D \subseteq V_1$ (resp. $V_2$), $A^1_D$ (resp. $A^2_D$) = ``$|D \cup Odd(D)| \leq cn$".
 
 We want to compute $Pr(A^1_D)$ with $D \subseteq V_1$.
Let $|D| = d\nu$ for some $d \in \left] 0,1 \right]$.
 For any $u \in V_2$, $Pr(``u \in Odd(D)") = \frac{1}{2}$.
 Thus, $Pr(|Odd(D)| \leq x) = (\frac{1}{2})^{\nu} \sum_{k=0}^{x} {\nu \choose k} \leq \left( \frac{1}{2} \right)^\nu 2^{\nu H \left( \frac{x}{\nu } \right)}$ where $H : t \mapsto -t\log_2(t)-(1-t)\log_2(1-t)$ is the binary entropy function.
 Then, $Pr(A^1_D) = Pr(``|D \cup Odd(D)| \leq cn") = Pr(``|D| + |Odd(D)| \leq cn") = Pr(``|Odd(D)| \leq cn - |D|") \leq 2^{\nu  \left( H(2c-d)-1 \right)}$.
 
  Let $\sigma(A^1_D) = \frac{1}{r{\nu  \choose d\nu }}$ for some $r \in \mathbb{R}$ that will be chosen later.
 First, we verify that $Pr(A^1_D) \leq \sigma(A^1_D) \prod_{D' \in V_1, D'' \in V_2} (1 - \sigma(A^1_{D'}))(1 - \sigma(A^2_{D''}))$.
 The product of the right-hand side of the previous equation can be written $p = \prod_{|D'| = 1}^\nu  \left( 1 - \frac{1}{r{\nu  \choose |D'|}} \right)^{2{\nu  \choose |D'|}} = \left[ \prod_{|D'| = 1}^\nu  \left( 1 - \frac{1}{r{\nu  \choose |D'|}} \right)^{r{\nu  \choose |D'|}} \right]^{\frac{2}{r}}$.
 The function $f : x \mapsto \left( 1-\frac{1}{x} \right)^x$ verifies $f(x) \geq \frac{1}{4}$ when $x \geq 2$, therefore $p \geq \left( \frac{1}{4} \right)^{\nu *\frac{2}{r}} = 2^{-\frac{4\nu }{r}}$ for any $r \geq 2$.
 Thus, it is sufficient to have $2^{\nu  \left( H(2c-d)-1 \right)} \leq \frac{1}{r{\nu  \choose d\nu }} 2^{-\frac{4\nu }{r}}$.
 Rewriting this inequality gives $r{\nu  \choose d\nu } 2^{(2c-1)\nu -d\nu +\frac{4\nu }{r}} \leq 1$.
 Thanks to the bound ${\nu  \choose d\nu } \leq 2^{\nu H\left( \frac{d\nu }{\nu } \right)}$ and after applying the logarithm function and dividing by $\nu $, it is sufficient that $\frac{\log_2{r}}{\nu } + H(d)+H(2c-d)-1+\frac{4}{r} \leq 0$.
 Therefore, if we take $r=\nu $ and $\nu  \to +\infty$, the asymptotic condition on the value of $c$ is $H(d)+H(2c-d)-1 \leq 0$.
 Since this bound must be verified for all $d \in \left( 0,1 \right]$, it must be true for the value of $d$ for which the function $d \mapsto H(d)+H(2c-d)-1$ is minimum.
 Usual techniques show that the minimum is reached for $d = c$, and a numerical analysis shows that $c=0.110$ satisfies the condition $Pr(A^1_D) \leq \sigma(A^1_D) p$ for some $r \in \mathbb{R}$ and $\nu  > \nu_0$.
 A similar reasoning is used to prove $Pr(A^2_D) \leq \sigma(A^2_D) p$ for all $D \in V_2$.
 
 The conditions and the choice of the weights $\sigma(A^1_D)$ and $\sigma(A^1_D)$ allow us to use the Lov\'asz Local Lemma (Lemma \ref{LLL}), and we derive $Pr \left( \big\{ \overline{A^1_D}\ \big|\ D \in V_1 \big\}, \big\{\ \overline{A^2_D}\ \big|\ D \in V_2 \big\} \right) \geq p > 0$, which proves that $Pr \left(\delta_{loc}(G_B) \geq cn \right) > 0$ for any $c \leq 0.110$ and for $\nu>\nu_0$.
 Then there exists at least one bipartite graph $G_B$ of order $n$ such that $\delta_{loc}(G_B) \geq 0.110n$.
 \hfill $\Box$
 \end{proof}

The general case of a random graph without the bipartite constraint leads to a slightly better constant:

 \begin{theorem}
 \label{G_1/2}
 There exists $n_0 \in \mathbb{N}$ such that for all $n > n_0$ there exists a graph of order $n$ whose local minimum degree is greater than $0.189n$.
 \end{theorem}
 
 Due to its similarity to the above proof, the proof of this theorem is given in Appendix.

\section{NP-completeness of the local minimum degree problem}

In this section, we show that given a graph $G$ and an integer $d$, deciding whether $\delta_{loc}(G) \leq d$ is NP-complete even for the family of bipartite graphs.
This result is established through a reduction to the problem of the shortest word of a linear code \cite{shortest_codeword_NPC} and uses the families of graphs whose existence has been proven in the previous section.

\begin{lemma}
\label{delta_loc_biparti}
Let $G=(V,E)$ be a bipartite graph.
Let $V=V_1 \cup V_2$ where $V_1$ and $V_2$ are the two parties of the graph $G$.
There exists $D_0 \subseteq V$ such that $\delta_{loc}(G) + 1 = |D_0 \cup Odd(D_0)|$ and $D_0 \subseteq V_1$ or $D_0 \subseteq V_2$.
\end{lemma}

\begin{proof}
Let $D \subseteq V$ such that $|D \cup Odd(D)| = \delta_{loc}(G) + 1$.
We write $D = D_1 \cup D_2$ with $D_1 \subseteq V_1$ and $D_2 \subseteq V_2$.
$D \neq \varnothing$, then without loss of generality, we assume that $D_1 \neq \varnothing$.
$G$ is bipartite, then $Odd(D_1) \subseteq V_2$ and $Odd(D_2) \subseteq V_1$.
Thus $Odd(D_1 \cup D_2) = Odd(D_1) \cup Odd(D_2)$, and $\delta_{loc}(G) + 1 = |D \cup Odd(D)| = |D_1 \cup Odd(D_1) \cup D_2 \cup Odd(D_2)| \geq |D_1 \cup Odd(D_1)| \geq \delta_{loc}(G) + 1$.
The bounds are tight, therefore $|D_1 \cup Odd(D_1)| + 1 = \delta_{loc}(G)$.
\hfill $\Box$
\end{proof}

\begin{theorem}
\label{delta_loc_NPC}
Given a graph $G$ and an integer $d$, deciding whether $\delta_{loc}(G) \leq d$ is NP-complete for the family of bipartite graphs.
\end{theorem}

\begin{proof}
The problem is in NP since a set of the form $D \cup Odd(D)$ with $D \neq \varnothing$ and $ |D \cup Odd(D)| = \delta_{loc}$ is a $YES$ certificate.
We do a reduction to the problem of the shortest codeword.
Let $A \in \mathcal{M}_{n+k,k}(\mathbb{F}_2)$ be the generating matrix of a binary code.
Using oracle for the problem related to the quantity $\delta_{loc}$ on bipartite graphs, we answer the problem of finding the shortest word of $A$.

If $dim(Ker(A)) \neq 0$ then $\min_{X \in \mathbb{F}_2^k, X \neq 0} \{ w(AX) \}$ = $0$, where $w$ is the Hamming weight function.
Otherwise, $\min_{X \in \mathbb{F}_2^k, X \neq 0} \{ w(AX) \} = \min_{X \in \mathbb{F}_2^k, X \neq 0} \{ w(X) + w(A'X) \}$ where $A$ is written in the form
$
\left(
   \begin{matrix}
      I_k \\
      A' \\
   \end{matrix}
\right)
$.
Thus, $A'$ is of size $n \times k$.

We want to construct a bipartite graph $G$ (Figure \ref{compose_bipartites}) on which the oracle call is performed.
To this purpose, we build two auxiliary graphs $G_{A'}$ and $G_B$ in a first time.
Let $G_{A'} = (V_{A'_1} \cup V_{A'_2}, E_{A'})$ be the bipartite graph defined as follows: the sets $V_{A'_1}$ of size $k$ and $V_{A'_2}$ of size $n$ denote both sides of the bipartition of $G_{A'}$, and for all $x \in V_{A'_1}$ and $x' \in V_{A'_2}$, $(x,x') \in E_{A'}$ if and only if $A'_{x',x} = 1$.
After that, thanks to Theorem \ref{linear_bipartite}, there exists $n_0 \in \mathbb{N}$ such that for all $n > n_0$ there exists a bipartite graph $G_B = (V_{B_1} \cup V_{B_2}, E_B)$ of order $10(n+1)$ such that $\delta_{loc}(G_B) > n+1$.
The sets $V_{B_1}$ and $V_{B_2}$ denote both sides of the bipartition of $G_B$.
Let $u$ be any vertex of $V_{B_1}$.
Consider the bipartite graph $G=(V_1 \cup V_2,E)$ (Figure \ref{compose_bipartites}) defined as follows: $V_1 = V_{1L} \cup V_{1R}$ with $V_{1L} = V_{A'_1} \times \left\{ u \right\}$ and $V_{1R} = V_{A'_2} \times V_{B_2}$, and $V_2 = V_{A'_2} \times V_{B_1}$.
For all $(x,y) \in V_1$ and $(x',y') \in V_2$, $\big( (x,y) , (x',y') \big) \in E$ if and only if $\big( (x,x') \in E_{A'} \wedge y = y' \big) \vee \big( (y,y') \in E_B \wedge x = x' \big)$.

Both independent sets $V_1$ and $V_2$ form a partition of the vertices of the graph.
Thanks to Lemma \ref{delta_loc_biparti}, there exists a non-empty set $D_0 \subseteq V(G)$ such that $\delta_{loc}(G) + 1 = |D_0 \cup Odd(D_0)|$ and $D_0 \subseteq V_1$ or $D_0 \subseteq V_2$.

Suppose that $D_0 \subseteq V_2$.
Therefore $\delta_{loc}(G) = |D_0 \cup Odd(D_0)| - 1 \geq \delta_{loc}(G_B) > n+1 \geq \delta(G)+1 \geq  \delta_{loc}(G)$.
This leads to a contradiction, therefore $D_0 \subseteq V_1$.

Suppose that $D_0 \cap V_{1R} \neq \varnothing$.
Let $v \in D_0 \cap V_{1R}$.
Then $\delta_{loc}(G) = |D_0 \cup Odd(D_0)| - 1 \geq |\{v\} \cup Odd(\{v\})| - 1 \geq \delta_{loc}(G_B) > n+1 \geq \delta(G)+1 \geq  \delta_{loc}(G)$.
This also leads to a contradiction, therefore $D_0 \subseteq V_{1L}$.

The reader will notice that since $D_0 \subseteq V_{1L}$, $|Odd(D_0)|$ in the graph $G$ can be written $w(A'X_{D_0})$ where $X_{D_0}$ is the vector representation of the set $D_0$.
Moreover, since $V_{1L}$ is an independent set, $|D_0 \cup Odd(D_0)| = |D_0| + |Odd(D_0)| = w(X_{D_0}) + w(A'X_{D_0})$.
By definition of $D_0$, we have $\delta_{loc}(G) + 1 = \min_{X \in \mathbb{F}_2^k, X \neq 0} \{ w(AX) \}$, which ends the reduction to the shortest codeword problem which is NP-complete \cite{shortest_codeword_NPC}.
\hfill $\Box$
\end{proof}

Notice that a constructive version of $NP$-completeness on non-necessarily bipartite graphs can be done by replacing the graph $G_B$ by a Paley graph in the above reduction.

\begin{figure}
\begin{center}
\begin{tikzpicture}
	\draw (0,0) ellipse (2mm and 1cm);
	\draw (2,0) ellipse (3mm and 2.2cm);
	\draw (0.1,1) -- (1.8,1.9) --  (0.1,-1) -- (1.8,-1.9) -- cycle;
	\draw (1.8,2) -- (2,1.8) -- (3,2.8) -- (2.8,3) -- cycle;
	\draw (3.8,2) -- (4,1.8) -- (5,2.8) -- (4.8,3) -- cycle;
	\draw (3.1,2.8) -- (4.5,2.8) --  (2.3,2) -- (3.7,2) -- cycle;
	\draw (1.8,-2) -- (2,-2.2) -- (3,-1.2) -- (2.8,-1) -- cycle;
	\draw (3.8,-2) -- (4,-2.2) -- (5,-1.2) -- (4.8,-1) -- cycle;
	\draw (3.1,-1.2) -- (4.5,-1.2) --  (2.3,-2) -- (3.7,-2) -- cycle;
	\draw (-1,0) node{{\Large $G_{A'}$}} (5.5,2.3) node{{\Large $G_B$}} (5.5,-1.7) node{{\Large $G_B$}};
	\draw (2,2)  node{$\bullet$} (2,1.4)  node{$\vdots$} (2,-1.2)  node{$\vdots$} (2,-2) node{$\bullet$} (5.5,1.6) node{$\vdots$} (5.5,-0.8) node{$\vdots$};
	\draw (0,2) node{{\Large $V_{1L}$}} (2.2,3.5) node{{\Large $V_2$}} (4.5,3.8) node{{\Large $V_{1R}$}};
\end{tikzpicture}
\end{center}
\caption{Construction of the graph $G$ from the bipartite graph $G_{A'}$ (ellipses) and several copies of the bipartite graph $G_B$ (rectangles). $V_1 = V_{1L} \cup V_{1R}$}
\label{compose_bipartites}
\end{figure}
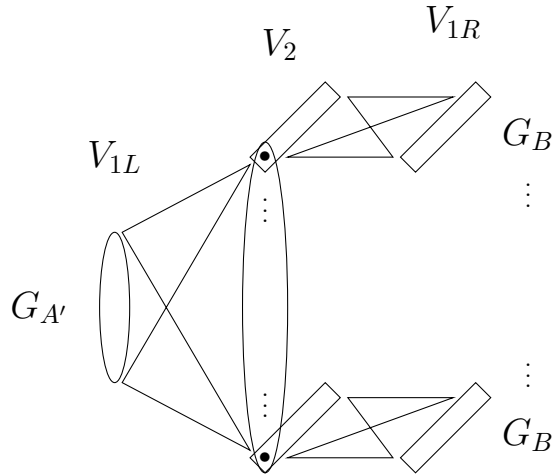

Since finding the local minimum degree is hard, one can wonder whether there exists a $l$-approximation algorithm for this problem for some constant $l$.
The previous reduction also shows that such an algorithm does not exist unless $P=NP$, even for the family of bipartite graphs.

\begin{theorem}
There exists no approximation algorithm with a constant factor for the problem of finding the local minimum degree of bipartite graphs, unless $P=NP$.
\end{theorem}

\begin{proof}
In the proof of Theorem \ref{delta_loc_NPC}, the value of $\delta_{loc}(G)$ where $G$ is constructed as described in Figure \ref{compose_bipartites} is the same as the shortest word of the linear code described by its generating matrix $A$.
This is true for any $A$, therefore for any constant $l$, any $l$-approximation of $\delta_{loc}(G)$ is a $l$-approximation of the Hamming weight of the shortest word of $A$.
Under the hypothesis $P \neq NP$, since finding the shortest codeword of a linear code is known to have no approximation algorithm with a constant factor \cite{approx_min_distance}, there exists no polynomial approximation algorithm with a constant factor for the problem of finding the local minimum degree of bipartite graphs.
\hfill $\Box$
\end{proof}

\section{Conclusion}

After having shown that the local minimum degree of the family of Paley graphs is greater than the square root of their order, we proved that there exist an infinite family of graphs whose local minimum degree is linear in their order (with constant at least $0.189$ in general and $0.110$ for bipartite graphs).
Then, a study of the computational complexity of the decision problem associated with $\delta_{loc}$ with a polynomial reduction to the problem of the shortest word of a linear code shows its NP-completeness, even on bipartite graphs.
It is also impossible to find an approximation algorithm with any constant factor for this problem, unless $P=NP$.
The specificity of the reduction performed lies in the fact that the construction of an instance for the problem associated with $\delta_{loc}$ uses the existence of a family of bipartite graphs proven above.
Thus, in a way, we proved that a polynomial reduction exists without constructing it explicitly.

Some questions remain open: is it possible to give an explicit family of graphs with linear local minimum degree?
Can we find a constructive proof of $NP$-completeness for the decision problem associated with $\delta_{loc}$ on bipartite graphs?
Can we find an infinite family of Paley graphs whose local minimum degree is linear?
The answer of the last question would provide an answer to the Bazzi-Mitter conjecture \cite{Bazzi-Mitter} on hyperelliptic curves.
 
\newcommand{\etalchar}[1]{$^{#1}$}

\appendix

\section{Proof of Theorem \ref{G_1/2}}

{\bf Theorem \ref{G_1/2}}
\emph{
 There exists $n_0 \in \mathbb{N}$ such that for all $n > n_0$ there exists a graph of order $n$ whose local minimum degree is greater than $0.189n$.
 }

\begin{proof}
Let $G$ be a graph of order $n$ where any possible edge exists with probability $\frac{1}{2}$.
We are looking for the greatest value of $c$ such that $Pr \left(\delta_{loc}(G) \geq cn \right) > 0$.
Thus, we want that ``$\forall D \subseteq V(G), |D \cup Odd(D)| > cn$".
Consequently, the events to avoid are $A_D$: ``$|D \cup Odd(D)| \leq cn$".
Obviously, it is sufficient to consider only the events $A_D$ with $D \leq cn$.

For all $D$ sucht that $|D| \leq cn$, we want to get an upper bound on $Pr(A_D)$.
Let $|D|=dn$ for some $d \in \left( 0,c \right] $.
For all $u \in V \setminus D$, $Pr(``u \in Odd(D)") = \frac{1}{2}$.
If $D$ is fixed, the events $``u \in Odd(D)"$ when $u$ is outside $D$ are independent.
Therefore, if the event $A_D$ is true, any but at most $(c-d)n$ vertices outside $D$ are contained in $Odd(D)$.
There are $(1-d)n$ vertices outside $D$, then $Pr(A_D) = \left( \frac{1}{2} \right)^{(1-d)n} \sum_{k=0}^{(c-d)n} {(1-d)n \choose k} \leq \left( \frac{1}{2} \right)^{(1-d)n} 2^{(1-d)nH\left( \frac{c-d}{1-d} \right)} = 2^{(1-d)n \left[ H\left( \frac{c-d}{1-d} \right) - 1 \right]}$ where $H : t \mapsto -t\log(t)-(1-t)\log(1-t)$ is the binary entropy function.

Let $\sigma(A_D) = \frac{1}{r{n\choose |D|}}$.
Let $p = \prod_{|D'| \leq cn} (1 - \sigma(A_{D'}))$.
In order to apply the L\'ovasz Local Lemma (Lemma \ref{LLL}), we want to have $Pr(A_D) \leq \sigma(A_D)p$.
The product $p$ verifies $p = \prod_{|D'| = 1}^{cn} \left(1 - \frac{1}{r{n\choose |D'|}} \right)^{{n \choose |D'|}} = \left[ \prod_{|D'| = 1}^{cn}  \left( 1 - \frac{1}{r{n  \choose |D'|}} \right)^{r{n  \choose |D'|}} \right]^{\frac{1}{r}}$.
The function $f : x \mapsto \left( 1-\frac{1}{x} \right)^x$ verifies $f(x) \geq \frac{1}{4}$ when $x \geq 2$, therefore $p \geq \left( \frac{1}{4} \right)^{\frac{cn}{r}} = 2^{-\frac{2cn }{r}}$ for any $r \geq 2$.
Thus, it is sufficient that $2^{(1-d)n \left[ H\left( \frac{c-d}{1-d} \right) - 1 \right]} \leq \frac{1}{r{n \choose dn}}2^{-\frac{2cn }{r}}$.
Rewriting this inequality with the bound ${n  \choose dn } \leq 2^{nH\left( \frac{dn }{n } \right)}$ and applying the logarithm function and dividing by $n$ gives the following sufficient condition: $(1-d) \left[ H \left( \frac{c-d}{1-d} \right) -1 \right] + H(d) + \frac{2c}{r} + \frac{\log_2{r}}{n} \leq 0$.
Taking $r=n$, the condition becomes asymptotically $(1-d) \left[ H \left( \frac{c-d}{1-d} \right) -1 \right] + H(d) \leq 0$.

Numerical analysis shows that this condition is true for any $c \leq 0.189$ and for all $d$ such that $0 < d \leq cn$.
Therefore, Lemma \ref{LLL} ensures that $Pr \left( \big\{\ \overline{A_D}\ \big|\ |D| \leq cn \big\} \right) \geq p > 0$, which proves the existence of at least one graph $G$ of order $n$ such that $\delta_{loc}(G) \geq 0.189n$.
This ends the proof of Theorem \ref{G_1/2}.
\hfill $\Box$
\end{proof}

\end{document}